%% file: GoCArcive.tex
\documentclass{article}

\PassOptionsToPackage{numbers}{natbib}

\usepackage[preprint]{nips_2018}

\input{mypackages.tex}

\input{mystyle.tex}

\title{Game of Coins}

\author{
Alexander Spiegelman\thanks{Alexander Spiegelman is
   grateful to the Azrieli Foundation for the award of an
   Azrieli Fellowship.} \\
Electrical Engineering\\
Technion IIT\\
Haifa, Israel \\
  \texttt{sashas@campus.technion.ac.il} \\
\And
Idit Keidar \\
Electrical Engineering \\
Technion IIT \\
Haifa, Israel \\
  \texttt{idish@ee.technion.ac.il}
\And
Moshe Tennenholtz\thanks{Moshe Tennenholtz is supported by EU
project 740435 - MDDS} \\
Industrial Engineering \\
Technion IIT \\
Haifa, Israel \\
  \texttt{moshet@ie.technion.ac.il}
}

\begin{document}

\maketitle


\input{abstract.tex}
\input{introduction.tex}
\input{model.tex}

\input{equilibriums.tex}

\input{betterEquilibrium.tex}

\input{ChangeEquilibrium.tex}

\input{discussion.tex}  

\newpage
\bibliographystyle{plain}
\bibliography{bibliography}

\newpage
\begin{appendices}
\input{appendix.tex}
\end{appendices}

\end{document}

%% file: mypackages.tex
\usepackage[utf8]{inputenc}
\usepackage{amsmath}
\usepackage{amssymb}
\usepackage{amsthm}
\usepackage[dvipdfm]{graphicx}
\usepackage[colorinlistoftodos]{todonotes}
\usepackage{algorithm}
\usepackage[noend]{algpseudocode}
\usepackage{enumerate}
\usepackage{xcolor}
\usepackage{framed}
\usepackage{bm}
\usepackage{pifont}
\usepackage{multicol}
\usepackage{lipsum}
\usepackage{tikz}
\usetikzlibrary{calc}
\usepackage{wrapfig}
\usepackage{array}
\usepackage{relsize}
\usepackage{comment}
\usepackage{url}
\usepackage[title]{appendix}
\usepackage{caption}
\usepackage{subcaption}

\usepackage{bmpsize}
\usepackage{natbib}

\usepackage[T1]{fontenc}    
\usepackage{booktabs}       
\usepackage{amsfonts}       
\usepackage{nicefrac}       
\usepackage{microtype}      

%% file: mystyle.tex
\newtheorem{observation}{Observation}
\newtheorem{claim}[observation]{Claim}
\newtheorem{lemma}{Lemma}

\newtheorem{definition}{Definition}
\newtheorem{assumption}{Assumption}
\newtheorem{theorem}{Theorem}
\newtheorem{proposition}{Proposition}

\DeclareSymbolFont{rmlargesymbols}{OMX}{mdbch}{m}{n}
\DeclareMathSymbol{\rmintop}{\mathop}{rmlargesymbols}{82}

\DeclareMathOperator*{\argmax}{argmax}

\newenvironment{claimclone}[1]{\noindent
\textbf{Claim~\ref{#1}} (restated).\em}{\par}

\newenvironment{lemmaclone}[1]{\noindent
\textbf{Lemma~\ref{#1}} (restated).\em}{\par}

\newenvironment{obsclone}[1]{\noindent
\textbf{Observation~\ref{#1}} (restated).\em}{\par}

%% file: abstract.tex
\begin{abstract}

We formalize the current practice of strategic mining in
multi-cryptocurrency markets as a game, and prove that any
better-response learning in such games converges to
equilibrium.
We then offer a reward design scheme that moves the system
configuration from any initial equilibrium to a desired one for
any better-response learning of the miners. Our work introduces
the first multi-coin strategic attack for adaptive and learning
miners, as well as the study of reward design in a multi-agent
system of learning agents.

\end{abstract}

%% file: introduction.tex
\section{Introduction}

Cryptocurrencies are an arms race.
Hundreds of digital coins have
crept into the worldwide market in the last
decade~\cite{MarketState}, including more than a dozen with over
a billion dollar Market Cap, e.g.,~\cite{buterin2014next,
BitcoinCash, Litecoin, Cardano, Neo}.
The vast majority of cryptocurrencies are based on the
notion of proof of work (PoW)~\cite{Bitcoin}.
As a result, the major strategic players in the context of
cryptocurrencies are {\em miners} who devote their power to
solving computational puzzles to find PoWs~\cite{Bitcoin,
buterin2014next}.

The miners for a particular coin usually gain rewards that are
proportional to the power they invest in the coin out of the
total invested power (in the coin) by all miners.
Each coin, therefore, can be viewed as having some \emph{weight}
that reflects the reward it divides among its miners.
In practice, a coin's weight (or reward) depends on its
transaction rate, transaction fees, and its fiat exchange rate.

While the above description is not complete, it does capture the
fundamental decision faced by the miner: where should I mine?
One indication for reward-based coin switching can be
found online in websites like
www.whattomine.com~\cite{whattomine}, where miners enter their
mining parameters (technology, power, cost, et cetra) and get a
list of coins they can mine for, ordered by their profitability.
Another interesting example happened on November 12
(2017)~\cite{bitinfocharts}, when a dramatic change in the
Bitcoin to Bitcoin Cash~\cite{BitcoinCash} (a spin-off from
Bitcoin) exchange rate led to a major inrush of miners from
Bitcoin to Bitcoin Cash (see Figure~\ref{figure:inrush}).

\begin{figure}[th]
    \centering
    \begin{subfigure}[t]{0.49\textwidth}
        \centering
        \includegraphics[width=0.95\textwidth]{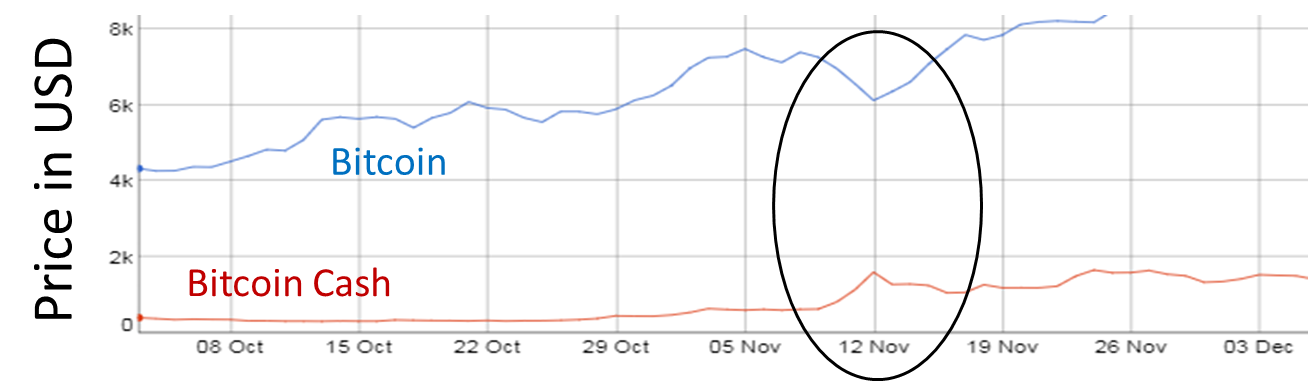}
        \caption{Bitcoin and Bitcoin Cash exchange rates over
        time.}
    \end{subfigure}%
    ~
    \begin{subfigure}[t]{0.49\textwidth}
        \centering
        \includegraphics[width=0.95\textwidth]{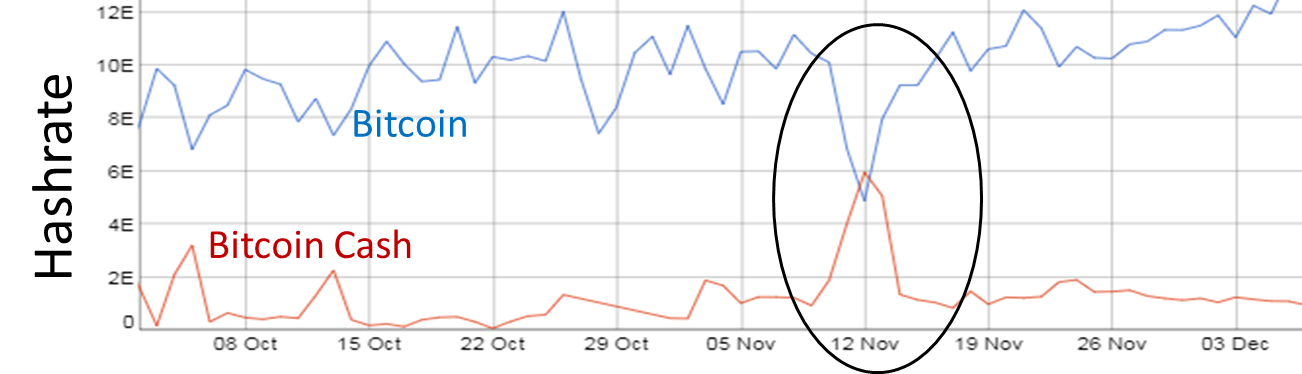}
        \caption{Hashrate corresponds to the number of
        miners.}
    \end{subfigure}
    \caption{Miners move from Bitcoin to Bitcoin Cash.}
    \label{figure:inrush}
\end{figure}

All in all, the structure of the cryptocurrency market suggests
that we face here a game among miners, where each miner wishes
to mine coins of heavy weights while avoiding competition with
other miners.
In this paper we introduce for the first time the study of the
cryptocurrency market as a game, consisting of a set of
strategic players (miners) with possibly different mining powers
and a set of coins with possibly different rewards (weights).
The miners are free to choose to mine for any coin from the set,
and we consider general better-response learning of
the miners.
That is, whenever any miner may benefit from deviating
(i.e., changing the coin it mines for), some miner will take a
step that improves his payoff; we allow an arbitrary sequence
of such individual improvement steps (sometimes called improving
path~\cite{MondererShapley96}).
In our first major result we prove that any such better response
learning converges to a (pure) equilibrium regardless of miner
powers and coin rewards!
This result is obtained by showing an \emph{ordinal potential},
which according to~\cite{MondererShapley96}, implies that
arbitrary better response learning converges to equilibrium.

Having at hand the above fundamental result, we move to a
discussion of strategic
manipulation~\cite{Nisan:2007:AGT:1296179}.
While many efforts have been invested in the study of
crypto-related manipulations~\cite{SelfishMining,
OptimalSelfishMining, StubbornMining, MinersDilemma}, we
introduce for the first time the manipulation of the miners' learning and optimization process.
Given that a shift in the weight of a coin
may influence miner behavior~\cite{whattomine}, in the
cryptocurrency setting, it is quite possible for an interested
party to affect this weight, either by
creating additional transactions with high fees (sometimes called whale
transactions~\cite{liao2017incentivizing}) or by manipulating
the coin exchange rate~\cite{gandal2018price,
priceMuniplautionm, whalesMuniplaution, BitcoinMuniplaution}.
This way, a miner (or another interested party) can attempt to
change the system equilibrium to a better one for them.
We show that under broad circumstances, for every
equilibrium of such a game, there exists a miner and another
equilibrium in which the miner's payoff is higher.
The question is therefore: can one design rewards (i.e.,
temporarily increase coin weights) in a way that will lead the
system from a given equilibrium to a desired one, so that the
system will remain in the desired equilibrium after reverting to
the original weights? Note that such reward design allows the
manipulator to pay a finite cost while gaining an advantage
indefinitely.

The above reward design problem is challenging since miners
might take \emph{any} better response step, and may make their
moves in any order. 
Given the (modified) weights, we can use our previous major
result to claim that any better response learning will converge
to an equilibrium.
Notice that the latter may not be the desired
one, but now we can modify the rewards again.
In the second major result of this paper we show that such
desired reward design for learning agents is feasible!
Namely, we provide a (multi-step) algorithm for assigning
rewards in equilibrium states that moves learning agents from any
initial equilibrium to a desired one.

In summary, our contributions are as follows:
\begin{enumerate}
  \item We formalize strategic mining in multi-cryptocurrency
  markets as a game (Section~\ref{sec:model}).
  
  \item We prove that any better-response learning in such games,
  starting from an arbitrary configuration, converges to
  equilibrium (Section~\ref{sec:equilibriums}).
  
  \item We show that, in many cases, for every equilibrium
  there is a miner and another equilibrium in which the miner's
  payoff is higher (Section~\ref{sec:betterEquilibrium}).

  \item We offer a reward design scheme that moves the system
  configuration from any initial equilibrium to a desired
  one for any better-response learning of the miners
  (Section~\ref{sec:movingEquilibrium}).
\end{enumerate}
For space limitations, the proofs of some of the claims we
state here are deferred to Appendices~C - F.


\subsection{Related work}

Results on better response learning convergence to pure
equilibrium are rare and are typically restricted to games with
exact potential~\cite{NIPS2017_7216,MondererShapley96}, which
coincide with congestion games.
We show that our game does not have an exact potential
(Section~\ref{sec:equilibriums}), and in fact our game belongs
to the larger class of ID congestion games, where the payoff of
a player depends on the player and the identity of other players
who choose a similar resource, rather than on their number only.
While there exist extensions of congestion games in which
better-response learning converges to equilibrium (e.g., a
restricted form of player-specific congestion games
\cite{Milchtaich96}, which does not include our game), such
results are extremely rare in the context of ID congestion
games.

Unlike works on learning in games that emphasize adapting
specific machine learning algorithms to minimize
regret~\cite{NIPS2017_7216, PalaiopanosPP17, NIPS2015_5763,
learninggamesbook}, we assume minimal rationality on behalf of
the players, i.e., that they follow an arbitrary
better response step improving their individual payoffs.

Our work also expands literature on reward
design~\cite{NIPS2010_4146,NIPS2017_7253,SorgSL10}, and
to the best of our knowledge, is the first to introduce reward
design for learning agents in a multi-agent setting. 
While seminal works in reward design assign/modify state rewards
in a reinforcement learning
context~\cite{DBLP:books/lib/SuttonB98}, we design rewards for
equilibrium states for \emph{any} better response learning.

Though several previous works presented game theoretical
analyses for cryptocurrencies~\cite{liao2017incentivizing,
SelfishMining, OptimalSelfishMining, StubbornMining,
carlsten2016instability, MiningGames, Johnson14, MinersDilemma,
Schrijvers16}, the vast majority of them deal (in one way or
another) with miners' incentives to follow the coins' mining
protocols.
Our work is the first to extend the study to a multi-coin
setting and establish fundamental game theoretical results
therein.

%% file: model.tex
\section{Model}
\label{sec:model}

A system in our model is a tuple $\langle \Pi, C \rangle$,
where $\Pi$ is a finite set of $n$ miners
(players) and $C$ is a finite set of coins
(resources).
A miner $p \in \Pi$ has mining power $m_{p} \in
\mathbb{R_+}$, which it can invest in one of the
coins
, i.e., the
set of possible actions of $p$ is $C$.
We denote the \emph{set of configurations} of a system $Q
=\langle \Pi, C \rangle$ as $S_Q \triangleq C^n$ and denote by
$s.p$ the action of player $p \in \Pi$ in configuration $s \in
S_Q$.
When clear from the context, we omit the subscript indicating the
system and simply write $S$.
Given  $s \in S$ and $c \in C$, we denote by
$P_{c}(s) \subseteq \Pi$ the set of miners who mine for $c$ in 
$s$, i.e., $P_{c}(s) \triangleq ~ \{p \in \Pi \mid s.p =
c \}$, and by $M_{c}(s)$ their total mining power, i.e.,
$M_{c}(s) \triangleq \Sigma_{p \in P_{c}(s)} m_{p}$.
For $s \in S_Q, ~p \in \Pi, ~c \in C$ we denote by $(s_{-p},
c)$ the configuration that is identical to $s$ except that
$s.p$ is replaced by $c$. 

A \emph{reward function} $F:C \to \mathbb{R}_+$
maps coins to rewards.
A game $G_{\Pi,C,F}$ consists of a system $\langle \Pi, C
\rangle$ and a reward function $F$.
Every coin in a game $G_{\Pi,C,F}$ divides its reward among all
the players that mine for it, and the miners' payoffs are
defined as follows:
For $s \in S$, the \emph{revenue per
unit (RPU)} of coin $c$ in $s$ is $RPU_{c}(G_{\Pi,C,F})(s)
\triangleq \frac{F(c)}{M_{c}(s)}$.
When clear from the context, we omit the the parameter indicating
the game.
The payoff function of a miner $p \in \Pi$ is $u_{p}(s)
\triangleq m_{p} \frac{F(s.p)}{M_{s.p}(s)} = m_{p} \cdot
RPU_{s.p}(s)$.

Given a game $G_{\Pi,C,F}$, a configuration $s \in S$, a miner $p
\in \Pi$, and a coin $c \in C$, we say that $p$ \emph{moves}
from $s.p$ to $c$ in $s$ if it changes its action from $s.p$ to
$c$.
A move from $s.p$ to $c$ is a \emph{better response step} for $p$  
if $u_{p}(s) < u_{p}((s_{-p}, c))$.
We say that a miner $p \in \Pi$ is \emph{stable} in a
configuration $s$ in game $G_{\Pi,C,F}$ if $p$ has no better
response steps in $s$.
A configuration $s$ is \emph{stable} or a (pure)
\emph{equilibrium} if every miner $p \in \Pi$ is stable in $s$.
A \emph{better response learning} from $s$ in $G_{\Pi,C,F}$ is a
sequence of configurations resulting from a sequence of better
response steps starting from $s$, which is either infinite or
ends with a stable configuration. 
In case it is finite, we say that it \emph{converges} to its
final configuration. 

A function $f : S \rightarrow \mathbb{R}$ is an \emph{ordinal
potential} for a game $G_{\Pi,C,F}$ if for any two configurations $s,s'
\in S$ s.t.\ some better response step of a miner $p \in \Pi$
leads from $s$ to $s'$, it holds that $f(s) < f(s')$.
If, in addition, $f(s') - f(s) = u_{p}(s') - u_{p}(s)$, then $f$
is an \emph{exact potential}.
By~\cite{MondererShapley96}, if a game $G_{\Pi,C,F}$ has an
ordinal potential, then every better response learning
converges.

%% file: equilibriums.tex
\section{Better response learning convergance}
\label{sec:equilibriums}

In this section we prove that although a game $G_{\Pi,C,F}$ has
no exact potential, every better-response learning of
the miners in game $G_{\Pi,C,F}$
converges to a stable configuration (pure equilibrium)
regardless of the sets $\Pi$ and $C$ and the reward function
$F$.
To gain intuition, the reader is referred to our
Appendix~A~and~B,
where we show how to construct a particular equilibrium in a game
$G_{\Pi,C,F}$ for any $\Pi,C$ and $F$, and give a simple ordinal
potential function for the symmetric case in which $F$ is a
constant function i.e., $\forall c,c' \in C$, $F(c) = F(c')$,
respectively.

\paragraph{No exact potential.} 
We start by showing that our game does not have an exact
potential.

\begin{proposition}

The game $G_{\Pi,C,F}$ does not always have an exact potential.

\end{proposition}

\begin{proof}

Let $G_{\Pi,C,F}$ be a game where $\Pi = \{p_1, p_2\}$,
$m_{p_1} =2, m_{p_2} = 1$, $C = \{c_1,c_2\}$, and $F(c_1) =
F(c_2) = 1$.
Assume by way of contradiction that $G_{\Pi,C,F}$ has an exact
potential function $H$, and consider the following four
configurations:
\begin{itemize}
  
  \item $s^1 = \langle c_1, c_1 \rangle$. Payoffs: $u_{p_1}(s^1)
  = \frac{F(c_1)\cdot m_{p_1}}{m_{p_1} + m_{p_2}} = 2/3$, 
  $u_{p_2}(s^1) = \frac{F(c_1)\cdot m_{p_2}}{m_{p_1} + m_{p_2}} =
  1/3$.
  
  \item $s^2 = \langle c_1, c_2 \rangle$. Payoffs:
  $u_{p_1}(s^2) = \frac{F(c_1)\cdot m_{p_1}}{m_{p_1}} = 1$, 
  $u_{p_2}(s^2) = \frac{F(c_2)\cdot m_{p_2}}{ m_{p_2}} = 1$.
  
  \item $s^3 = \langle c_2, c_2 \rangle$. Payoffs: $u_{p_1}(s^3)
  = \frac{F(c_2)\cdot m_{p_1}}{m_{p_1} + m_{p_2}} = 2/3$, 
  $u_{p_2}(s^3) = \frac{F(c_2)\cdot m_{p_2}}{m_{p_1} + m_{p_2}} =
  1/3$.
  
  \item $s^4 = \langle c_2, c_1 \rangle$. Payoffs: $u_{p_1}(s^4)
  = \frac{F(c_2)\cdot m_{p_1}}{m_{p_1}} = 1$, 
  $u_{p_2}(s^4) = \frac{F(c_1)\cdot m_{p_2}}{ m_{p_2}} = 1$.
  
\end{itemize}
Note that $H(s_2) - H(s_1) + H(s_3) - H(s_2) + H(s_4) -
H(s_3) + H(s_1) - H(s_4) = 0$.
However, by definition of exact potential, we get that 
$(H(s_2) - H(s_1))
+ (H(s_3) - H(s_2)) + (H(s_4) - H(s_3)) + (H(s_1) - H(s_4)) = 
(2/3) + (-1/3) + (2/3) + (-1/3) = 2/3 \neq 0$.
A contradiction.



\end{proof}


\paragraph{Ordinal potential.}
To show an ordinal potential, we use the following definitions:
%

For a configuration $s \in S$ in a game $G_{\Pi,C,F}$, we define
$list(s)$ to be the sequence of pairs in  $\{ \langle
RPU_c(s),c \rangle \mid c \in C\}$ ordered lexicographically
from smallest to largest.
Denote by $v_i(s)$ the coin (second element of the pair) in the
$i^{th}$ entry in $list(s)$.
Consider the ordered set $\langle L ,\prec_L \rangle$, where $L
\triangleq \{list(s) \mid s \in S\}$ is the set of all possible
lists in $G_{\Pi,C,F}$, and $\prec_L$ is the lexicographical
order.
The \emph{rank} of a list $list(s) \in L$, $rank(list(s))$, is
the rank of $list(s)$ in $\prec_L$ from smallest to largest.

%
%

Note that since $\Pi$ and $C$ are finite, we know that $S$ and
$L$ are finite.
The following two observations
establish a connection between better response steps and the
$RPUs$ of the associated coins.

 
\begin{observation}
\label{ob:neverMoveBack}

Consider a game $G_{\Pi,C,F}$, $s \in S$, 
$v_i(s) \in C$, and $p \in \Pi$ s.t.\ $s.p=v_i(s)$.
Then in every better response step of $p$ that changes $s.p$ to
a coin $v_j(s)$, it holds that $j > i$.

\end{observation}

\begin{observation}
\label{ob:movingImproving}

Consider a game $G_{\Pi,C,F}$.
If some better response step from configuration $s$ to
configuration $s'$ of a miner $p$ changes $s.p = c$ to
$s'.p = c'$, then $RPU_c(s) < min(RPU_c(s'), RPU_{c'}(s'))$. 

\end{observation}

\noindent We are now ready to prove that any game $G_{\Pi,C,F}$
has an ordinal potential function.

\begin{theorem}
\label{theorm:generalPotential}

For any finite sets $\Pi$ and $C$ of miners and coins and 
reward function $F$, $H(s) \triangleq rank(list(s))$ is an
ordinal potential in the game $G_{\Pi,C,F}$.

\end{theorem}

\begin{proof}

Consider two configurations $s,s' \in S$ s.t.\ some better
response step of a miner $p \in \Pi$ leads from configuration
$s$ to configuration $s'$, and let $v_i(s) = s.p$ and $v_j(s)
=s'.p$.
We need to show that $H(s) < H(s')$. 
Since only the RPUs of $v_i(s)$ and $v_j(s)$ are affected we get
that
\begin{equation}
\label{eq:foralleq}
\forall k \neq i,j~ RPU_{v_k(s)}(s) =
RPU_{v_k(s)}(s').
\end{equation}
By Observation~\ref{ob:neverMoveBack}, we get that
$j> i$, and thus $\forall k, ~1 \leq k < i$, $
RPU_{v_k(s)}(s) = RPU_{v_k(s)}(s').$
By Observation~\ref{ob:movingImproving}, we get
that $min(RPU_{v_i(s)}(s'), RPU_{v_j(s)}(s')) > RPU_{v_i(s)}(s)$,
and thus, together with the definition of $v_i$ and
Equation~\ref{eq:foralleq}, we get that $\forall k, ~i \leq k \leq |C|$, $ RPU_{v_k(s)}(s') \geq
RPU_{v_i(s)}(s).$
Therefore, none of them ``move down'' to a position before $i$ in
$list(s')$ and so 
\begin{equation}
\label{eq:smallereq}
 \forall k, ~1 \leq k < i, 
\langle RPU_{v_k(s)}(s), v_k(s)\rangle = 
 \langle RPU_{v_k(s')}(s'), v_k(s')\rangle.
\end{equation}
That is, the first $i-1$ elements of $list(s)$ are equal to
the first $i-1$ elements of $list(s')$.
Hence, it suffices to show that the $i^{th}$ element of
$list(s')$ is lexicographically larger than the $i^{th}$ element
of $list(s)$.
Let  $v_l(s) = v_i(s')$. From Equation~\ref{eq:smallereq}, we
know that $l \geq i$, so there are two possible cases:
\begin{itemize}
  
  \item First, $l \in \{i, j\}$. The theorem follows from
  Observation~\ref{ob:movingImproving}.

   \item Second, $l > i,~ l\neq j$. In this case, 
   \begin{align*}
   \langle RPU_{v_i(s)}(s), v_i(s) \rangle & < 
   \langle RPU_{v_l(s)}(s), v_l(s) \rangle &
   (\text{lexicographical order of $list(s)$})\\
   &=
   \langle RPU_{v_l(s)}(s'), v_l(s) \rangle &
   (\text{Equation~\ref{eq:foralleq}})\\
   &=
   \langle RPU_{v_i(s')}(s'), v_i(s') \rangle &
   (\text{Definition of $v_l(s)$}),
   \end{align*}
   as needed.

\end{itemize}

\end{proof}

%% file: betterEquilibrium.tex
\section{There is often a better equilibrium}
\label{sec:betterEquilibrium}

Before moving to our second major result in which we describe
a manipulation through dynamic reward design that transitions
the system between equilibria, in this section we show that
under broad circumstances, in every stable configuration there
is at least one miner who has higher payoff in another stable
configuration. This means that such a miner will gain from
moving the system there. 
Specifically, we prove this for games that satisfy the following
assumptions (note that we use these assumptions only in this
section):

\begin{assumption}[Never alone]
\label{ass:alone}

For a configuration $s \in S$ in a game $G_{\Pi,C,F}$,
if there is a coin $c \in C$ s.t.\
$|P_{c}(s)| \leq 1$, then there is a miner $p \in \Pi$ s.t.\
changing $s.p$ to $c$ is a better response step for $p$.

\end{assumption}

\noindent Although this assumption cannot hold when $|\Pi| <
2|C|$, it often holds in practice since the number of
miners must be much larger than the number of coins for the
cryptocurrency to be secure (truly decentralized).

\begin{assumption}[Generic game]
\label{ass:genric}

For any two coins $c \neq c' \in C$ and two sets of players
$P, P' \subseteq \Pi$ in a game $G_{\Pi,C,F}$,
$\frac{F(c)}{\Sigma_{p \in P} m_{p}} \neq
\frac{F(c')}{\Sigma_{p \in P'} m_{p}}$.



\end{assumption}

\noindent This assumption is common in game theory
~\cite{DBLP:journals/mss/HolzmanL03}, and it makes sense in our
game since mining power in practice is measured in billions of
operations per hour and coin rewards are coupled with coin fiat
exchange rates, so exact equality is unlikely.

The following observation follows from
Assumption~\ref{ass:alone} and the fact that coins that are chosen by at
least one miner always divide their entire reward.
It stipulates that in every stable configuration, the sum of the
payoffs the miners get is equal to the sum of the coins' rewards.

\begin{observation}[All stable configurations are globally
optimal]
\label{obs:optimal}

For every stable configuration $s \in S$ in a game $G_{\Pi,C,F}$
under Assumption~\ref{ass:alone}, it holds that
${\sum}_{p \in \Pi} u_{p}(s) = {\sum}_{c \in C} F(c)$.

\end{observation}

From Observation~\ref{obs:optimal} and
Assumption~\ref{ass:genric} it is easy to show the following
claim:  

\begin{claim}
\label{claim:higherpayoff}

Consider a game $G_{\Pi,C,F}$ under
Assumptions~\ref{ass:alone} and~\ref{ass:genric}.
If the game has more than one stable configuration,
then for every stable configuration $s$ there exist a
miner $p$ and a stable configuration $s'$ s.t.\ $u_{p}(s') >
u_{p}(s)$.

\end{claim}

It remains to show that $G_{\Pi,C,F}$ has more than one stable
configuration.
Consider $\Pi = \{p_1,\ldots,p_n\}$ s.t.\ $m_{p_1}
\geq m_{p_2} \geq \ldots \geq m_{p_n}$, and $\forall i, 1 \leq
i \leq n$, let $\Pi_i = \{p_1,\ldots,p_i\}$.
We first show that the game $G_{\Pi_2,C,F}$ has two different
configurations in which miners $p_1,p_2$ do not share a coin and
at most one of them is unstable.
Then, we inductively construct two configurations in
$G_{\Pi_i,C,F}$, $\forall i, 3 \leq i \leq n$, based on the two
configurations in $G_{\Pi_{i-1},C,F}$, in which all miners in
$\Pi_{i-1}$ keep their locations and all miners except maybe the
one that was unstable in $G_{\Pi_1,C,F}$ are stable.
The construction step is captured by Claim~\ref{claim:SRS},
where $p_{new} = p_i$ in the $i^{th}$ step.
%
%

\begin{claim}
\label{claim:SRS}

Let $F$ be a reward function.
Consider a system $Q= \langle \Pi,C \rangle$, and a
configuration $s \in S_Q$.
Now consider another system $Q'= \langle \Pi',C \rangle$ s.t.\
$\Pi' = \Pi \cup \{p_{new}\}$, $p_{new} \not\in \Pi$, and
$m_{p_{new}} \leq min\{m_{p} | p \in \Pi \}$.
Let $c = \underset{c' \in C} \argmax ~F(c')
\frac{m_{p_{new}}}{M_{c'}(s) + m_{p_{new}}}$ and consider a
configuration $s' \in S_{Q'}$ s.t.\ for all $p \in \Pi$
$s'.p=s.p$ and $s'.p_{new} = c$.
Then $p_{new}$ is stable in $s'$ in game $G_{\Pi',C,F}$, and
every player $p \in \Pi$ that is stable in $s$ in $G_{\Pi,C,F}$
is also stable in $s'$ in $G_{\Pi',C,F}$.

\end{claim}

Finally, we show that the two configurations we construct in
$G_{\Pi,C,F}$ are stable:
Let $p_{ns}$ be the (possibly) unstable miner.
By Assumption~\ref{ass:alone} (note that the assumption refers
only to game $G_{\Pi,C,F}$), $p_{ns}$ cannot be alone in a coin
(otherwise there must be another unstable miner), and thus it
shares the coin with a smaller stable miner, which we show
implies that $p_{ns}$ is stable.

%

\noindent Our results are captured by the following proposition,
which follows from Claim~\ref{claim:higherpayoff}
and and the inductive construction using Claim~\ref{claim:SRS}.

\begin{proposition}

Consider a game $G_{\Pi,C,F}$ under Assumptions~\ref{ass:alone}
and~\ref{ass:genric}. 
Then for every stable configuration $s$ in $G_{\Pi,C,F}$ there
exist a miner $p$ and a stable configuration
$s' \neq s$ in which $u_{p}(s') > u_{p}(s)$.

\end{proposition}

%% file: ChangeEquilibrium.tex
\section{Reward design: moving between equilibria}
\label{sec:movingEquilibrium}

In this section we consider a system $Q = \langle \Pi, C
\rangle$, where $\Pi = \{p_1,\ldots p_n\}$ s.t.\ $m_{p_1} >
m_{p_2} > \ldots > m_{p_n}$. 
For every reward function $F$ and every two stable
configurations $s_0,s_f \in S_Q$ in game $G_{\Pi,C,F}$ we
describe a mechanism to move the system from $s_0$ to the
desired configuration $s_f$ by temporarily increasing
coin rewards.
Note that once we lead the system to $s_f$, we can
return to the original rewards (i.e., stop manipulating coin
weights) because $s_f$ is stable in $G_{\Pi,C,F}$.
Therefore, a manipulator who gains from moving to a desired
stable configuration can do it with a bounded cost.

We first define a reward design function that maps
system configurations to reward functions.

\begin{definition}[reward design function]

Consider a system $Q$.
A reward design function $F$ for system $Q$ is a
function mapping every configuration $s \in S_Q$ to a reward
function, i.e., $F(s):C \to \mathbb{R}_+$. 

\end{definition}

\paragraph{Dynamic reward design.}

Consider a system $\langle \Pi,C \rangle$ and a reward function
$F$.
A dynamic reward design mechanism for game $G_{\Pi,C,F}$
is an algorithm that for any two stable configurations $s_0,s_f$
in $G_{\Pi,C,F}$ moves the system from $s_0$ to $s_f$ by
following the protocol in Algorithm~\ref{alg:DRDM}.

\begin{algorithm}
\caption{protocol to move a system $\langle \Pi,C \rangle$
with reward function $F$ from $s_0 $ to $s_f $.} 
\begin{algorithmic}[1]
\State   $s \leftarrow$ $s_0$

\Repeat
	\State choose a reward design function $H$ s.t.\ for all $c \in
	C$, $H(s)(c) \geq F(c)$
	\State allow better-response learning in $G_{\Pi,C,H(s)}$,
	starting from $s$, to converge to some stable
	\Statex \hspace*{4.7mm} configuration $s'$
	\Comment convergence is due to
	Theorem~\ref{theorm:generalPotential}
	
	\State $s \leftarrow$ $s'$
\Until{$s=s_f$} 
\end{algorithmic}
\label{alg:DRDM}
\end{algorithm}

\subsection{Reward design algorithm}

To describe a dynamic reward design algorithm we need to specify
the reward design function for every loop iteration in
Algorithm~\ref{alg:DRDM}.
Intuitively, we observe that miners with less mining power are
easily moved between coins, meaning that we can increase a coin
reward so that a \emph{small miner} with little mining
power will benefit from moving there, but \emph{bigger miners}
with more mining power prefer to stay in their current
locations.
Therefore, the idea is to evolve the current configuration to
$s_f \in S$ in $n = |\Pi|$ stages, where in stage $i$, we move
the $n -i +1$ miners with the smallest mining powers to the location
(coin) of miner $p_i$ in the final configuration $s_f$ (i.e.,
$s_f.p_i$) while keeping the remaining miners in their (final)
places.
To this end, we define $n$ intermediate configurations.
For $i$, $1 \leq i \leq n$, we define $s_i$ as:
\begin{equation}
\label{eq:stages}
s_i.p_k=
 \begin{cases}
 s_f.p_k & \forall 1 \leq k \leq i \\
 s_f.p_{i} & \forall i < k \leq n
\end{cases}
\end{equation}
That is, in $s_i$, miners $p_1,\ldots,p_i$ are in their final
locations and miners $p_i,\ldots,p_n$ are in the final location
of miner $p_i$. Note that $s_n = s_f$. Figure~\ref{fig:stagei}
illustrates the stage transitions in the algorithm.

\begin{figure}[th]
    \centering
    \begin{subfigure}[t]{0.44\textwidth}
        \centering
       \includegraphics[width=1.3in]{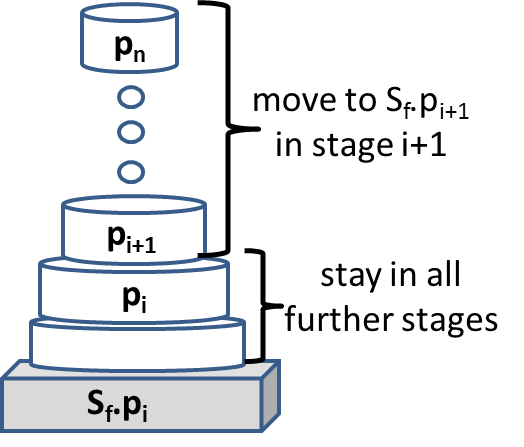}
        \caption{Configuration $s_i.$
        }
    \label{fig:stagei}
    \end{subfigure}%
    ~ 
    \begin{subfigure}[t]{0.44\textwidth}
        \centering
    \centering
        \includegraphics[width=2.3in]{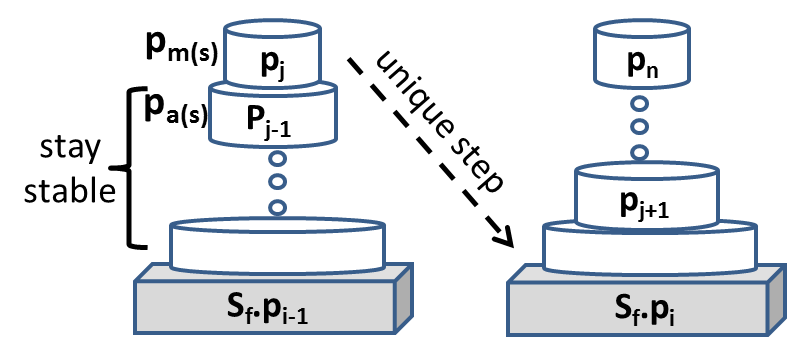}
        \caption{Iteration moving $p_{m_i(s)}$; $p_{a_i(s)}$ is
        the anchor.
        }
    \label{fig:anchor}
    \end{subfigure}
    \caption{Reward design algorithm: (a) stages; and (b)
    iteration within stage $i$.
    Boxes represent coins, discs represent miners. 
    The unlabeled bottom discs represent possible bigger miners
    who are in their final locations.}
    \label{fig:rewardChange}
\end{figure}

Notice that since we allow arbitrary better response learning
(in every iteration), choosing a reward design function is a
subtle task; miners can move according to any better response
step, and we cannot control the order in which miners move.
One may attempt to design a reward function so that in the
resulting game there is exactly one unstable miner with exactly
one better response step in the current configuration.
However, even given such a function, after that miner takes its
step, other miners might become unstable, which can in turn lead
to a learning process that depends on the order in which miners move and on the
choices they make (in case they have more than one better
response step).
Hence, the main challenge is to be able to restrict the set of
the possible stable configurations reached by learning
phase in each iteration.

In every loop iteration of stage $i > 1$ we pick a
miner $p_k$ that we want to move from $s_f.p_{i-1}$ to
$s_f.p_{i}$ (as explained shortly) and
choose the reward function carefully so that (1) $p_k$'s only better
response step is $s_f.p_i$, (2) all other miners are stable, and
(3) in every stable configuration reached by better response
learning after $p_k$'s step, $p_k$ is in $s_f.p_i$, all miners
$p_{k+1},\ldots, p_n$ are in either $s_f.p_{i-1}$ or $s_f.p_i$,
and all the other (bigger) miners remain in their (final)
locations.

Moreover, our proof shows by induction that our reward design
function of stage $i$ (defined below) guarantees that the set of
possible configurationas reached by learning in stage $i > 1$ is
\[ T_i \triangleq \{ s \in S  \mid (\forall k, ~1 \leq k \leq
i-1: ~s.p_k=s_f.p_k) \wedge (\forall k, i \leq k \leq n: ~s.p_k
\in \{s_f.p_i, s_f.p_{i-1})\}
\}.\]
Notice that the stage starts at $s_{i-1} \in T_i$. 
We now explain how we choose the reward design function for
stage $i$.
First, for every configuration $s \in T_i\setminus \{s_i\}$, the
index of the miner we want to move from $s_f.p_{i-1}$ to
$s_f.p_{i}$ (called \emph{mover}) is
\[
m_i(s)= min \{j | \forall l, ~j < l \leq n: s.p_l=s_f.p_i \}.
\]
Note that for every $s \in T_i$, $i \leq m_i(s) \leq n$.
Moreover, $p_{m_i(s)} \not\in P_{s_f.p_i}(s)$ and $m_i(s_{i-1})
= n$.
Let $a_i(s) = m_i(s) - 1$.
Intuitively, we use $p_{a_i(s)}$ as an \emph{anchor} in
configuration $s$; we choose a reward function that increases
the reward of coin $s_f.p_i$ as high as possible without making
the anchor unstable.
As a result, all the miners in $P_{s_f.p_{i-1}}(s)$ (who are
bigger than or equal to the anchor) remain stable, and miner
$p_{m_i(s)}$ has a unique better response step to move to $s_f.p_i$.
Figure~\ref{fig:anchor} illustrates $m_i(s)$ and $a_i(s)$ for
some configuration $s \in T_i$.

In order to make sure that miners not in
$P_{s_f.p_{i-1}}(s) \cup P_{s_f.p_{i}}(s)$ also remain stable,
and in order to guarantee that that every better response
learning after $p_{m_i(s)}$'s step converges to a configuration
in $T_i$, we choose a reward function that evens out the RPUs of
all coins other than $s_f.p_i$.
For $s \in S$, let $R(s) = max\{ RPU_{c}(s) \mid c \in C \}$
$\forall s \in S$. 
The reward design function $H_i$ for stage $i > 1$ is:
\begin{equation}
\label{eq:Hi}
\forall i>1 ~\forall s \in T_i ~\forall c \in C, ~H_i(s)(c) = 
   \begin{cases}
     R(s) \cdot (M_{c}(s) + m_{p_{a_i(s)}}) & \mbox{for}\    c =
     s_f.p_i\\
     R(s) \cdot M_{c}(s)	& \mbox{otherwise} 
\end{cases}
\end{equation}
Note that the RPUs of all coins except $s_f.p_i$ in the game
$G_{\Pi,C,H_i(s)}$ are equal to $R(s)$. 
In addition, note that if a miner bigger than or equal to
$p_{a_i(s)}$ moves to $s_f.p_i$, then $s_f.p_i$'s RPU becomes no
bigger than $R(s)$.
However, since $m_{p_{m_i(s)}} < m_{p_{a_i(s)}}$, $p_{m_i(s)}$
has a unique better response step to move to $s_f.p_i$.
Therefore, we get that our reward design function allows us to
control the first step of the learning process.
In the next section we give more intuition on how it also
restricts the stable configuration at the end of any learning
process at stage $i$ to the set $T_i$.


As for the fist stage, note that we need to move all miners
to coin $s_f.p_1$, so intuitively we only need to increase its
reward high enough. 
We therefore choose:
\begin{equation}
\label{eq:H1}
\forall s \in S ~\forall c \in C, H_1(s)(c) = 
\begin{cases}
	max\{ F(c') \mid c' \in C\}\cdot \Sigma_{p \in \Pi}
	m_p & \mbox{for}\    c = s_f.p_1\\
	F(c)	& \mbox{otherwise} 
\end{cases}
\end{equation}
In Algorithm~\ref{alg:DRDMAlgorithm} we present our reward
design algorithm, and in the next section we outline
the proof that every stage eventually completes.

\begin{algorithm}
\caption{Dynamic reward design algorithm.} 
\begin{algorithmic}[1]
\State   $s \leftarrow$ $s_0$

\For{i=1 \dots n} 
\Repeat 
\Comment{$H_i$ is defined in Equations~\ref{eq:Hi}
and~\ref{eq:H1}} \State allow better-response learning in
$G_{\Pi,C,H_i(s)}$, starting from $s$, to converge to some stable 
	\Statex \hspace*{1cm}  configuration $s'$
	
	\State $s \leftarrow$ $s'$
\Until{$s=s_i$}
\Comment{$s_i$ is defined in Equations~\ref{eq:stages}}
\EndFor

\end{algorithmic}
\label{alg:DRDMAlgorithm}
\end{algorithm}


\subsection{Proof outline}

The proof for stage 1 is straightforward so we skip it.
Consider stage $i > 1$.
We prove in the appendix the following technical
lemma about stable configurations in the stage:

\begin{lemma}
\label{lem:inv}

Consider a configuration $s \in T_i \setminus \{s_i\}$.
Then every better response learning in the game
$G_{\Pi,C,H_i(s)}$ that starts at $s$ converges to a
configuration $s' \in T_i$ such that:
\begin{enumerate}
  
  \item $\forall k, ~1 \leq k < m_i(s)$, $s'.p_k = s.p_k$. 
  
  \item $s'.p_{m_i(s)} = s_f.p_i$.  
  
\end{enumerate}
\end{lemma} 

As part of the proof, we show that within stage $i$, all the
reached configurations (both stable and unstable) are in $T_i$. 
Let $c=s_f.p_{i-1}$ and $c' = s_f.p_{i}$.
After $p_{m_i(s)}$ moves to $c'$ according to
its only better response step, in the resulting configuration
$s'$, the RPUs of all coins not in $\{c, c'\}$ remain $R(s)$.
Moreover, $RPU_{c}(s') = \frac{H_i(c)}{M_c(s')} = 
 \frac{R(s)\cdot
M_{c}(s)}{M_{c}(s) - m_{p_{m_i(s)}}}$, and $RPU_{c'}(s') =
\frac{H_i(c')}{M_{c'}(s')}=
\frac{R(s)\cdot (M_{c}(s) + m_{p_{a_i(s)}})}{M_{c}(s) +
m_{p_{m_i(s)}}}$.
Therefore, although $RPU_{c}(s') > RPU_{c}(s)$, it is still
not high enough to drive miners not in $P_{c'}(s')$ (by
definition, bigger than $p_{m_i(s)}$) to move to it.
So the only miners that possibly have better response steps at
$s'$ are miners in $P_{c'}(s')$ who wish to move to $c$.
Moreover, the total mining power of the miners who actually move
to $c$ is smaller than $p_{m_i(s)}$, otherwise, $c$'s RPU will
go below $R(s)$.
In the proof we use the above intuition to formulate an invariant
that captures the lemma statement and prove it by induction on
better response steps. The lemma then follows from
Theorem~\ref{theorm:generalPotential} (every better
response learning converges to a stable configuration).

We next use Lemma~\ref{lem:inv} to prove that every stage $i > 1$
completes in a finite number of loop iterations.
To this end, we associate with every configuration $s\ T_i$ a
binary vector $v(s)$ indicating, for each $j \geq i$, whether
$p_j$ is in $P_{s_f.p_i}(s)$, where it needs to be at the end of
the stage.
Consider the ordered set $\langle V, \prec_v \rangle$,
where $V \triangleq \{0,1\}^{n-i+1}$ is the set of all binary
vectors of length $n-i+1$, and $\prec_v$ is the lexicographical
order.
For a configuration $s \in T_i$, we define $vec(s)$ to be a
vector in $V$ such that: 
\[
\forall j, ~1 \leq j \leq n-i+1, ~vec(s)[j]=
\begin{cases}
1 & \mbox{if    } ~p_{j + i -1} \in P_{s_f.p_i}(s)\\
0 & \mbox{otherwise}
\end{cases}
\]
and the function $\Phi_i: T_i \to \{1,..,|V|\}$ to be the rank
of $vec(s)$ in $V$.

\begin{theorem}
\label{theorm:lex}

Every stage $i > 1$ of Algorithm~\ref{alg:DRDMAlgorithm} 
completes in a finite number of loop iterations.

\end{theorem}

\begin{proof}

By definitions of stage $i$ and set $T_i$, the first
configuration of stage $i$ is $s_{i-1} \in T_i$.
By definition of $m_i(s)$ , for all $s \in T_i \setminus
\{s_i\}$, $m_i(s) \not\in P_{s_f.p_i}(s)$.
By inductively applying Lemma~\ref{lem:inv}, we get that every
loop iteration in stage $i$ ends in a configuration in $T_i$.
Therefore, consider a loop iteration of stage $i$ that starts in
configuration $s \neq s_i$ and ends in configuration $s'$,
we get by Lemma~\ref{lem:inv} that $m_i(s) \in
P_{s_f.p_i}(s')$ and $\forall k, ~i \leq k < m_i(s)$, $s.p_k =
s'.p_k$.
Therefore $\Phi(s') > \Phi(s)$.
Now since the set $T_i$ is finite, we get the after a finite
number of iterations we reach configuration $s_i$.

\end{proof}

%% file: discussion.tex
\section{Discussion}
\label{sec:discussion}

Our work studies and challenges the crypocurrency market from a
novel angle -- the strategic selections by adaptive miners
among multiple coins. 
There are several central followups one may consider. 
First, our reward design is effective for arbitrary
better-response learning, but one may wonder about its
speed of convergence under specific markets. 
In addition, we consider convergence to equilibrium, and one may
consider also convergence to a bad (possibly unstable)
configuration in which, for example, a particular miner will
have a dominant position in a coin, killing (at least for a
while) the basic guarantee of non-manipulation (security) for
that coin and allowing him to get a bigger portion of the reward.
One also may wonder about the asymmetric case where some coins
can be mined only by a subset of the miners.

%% file: appendix.tex
\section{Existence of an equilibrium}
\label{app:equilibrium}

\label{sub:ExEq}

We show here how to find a stable configuration (pure
equilibrium) in the game $G_{\Pi,C,F}$ for any $\Pi,C$, and $F$.
We do this by induction, selecting coins for miners in
descending order of mining power.

\begin{claim}
\label{claim:eq}

Consider a reward function $F$, a system $Q= \langle
\Pi,C \rangle$, and another system $Q'= \langle \Pi',C
\rangle$ s.t.\ $\Pi' = \Pi \cup \{p_{new}\}$, $p_{new} \not\in
\Pi$, and $m_{p_{new}} \leq min\{m_{p} | p \in \Pi \}$.
Then, if the game $G_{\Pi,C,F}$ has a stable configuration,
than the game $G_{\Pi',C,F}$ has a stable configuration as
well.

\end{claim}

\begin{proof}

Let $s$ be a stable configuration in $G_{\Pi,C,F}$. 
We use it to build a configuration $s'$ in
$G_{\Pi',C,F}$ in the following way:
for all $p \neq p_{new}$ set $s'.p = s.p$, and set
$s'.p_{new}$ to $c = \underset{c' \in C} \argmax ~
F(c') \frac{m_{p_{new}}}{M_{c'}(s) + m_{p_{new}}}$.
We now show that configuration $s'$ is stable in
$G_{\Pi',C,F}$.

First, consider $p_{new}$. 
Since we pick $c$ to be $\underset{c' \in C} \argmax ~F(c')
\frac{m_{p_{new}}}{M_{c'}(s) + m_{p_{new}}}$, we get that
$\forall c' \in C$, $F(c) \frac{m_{p_{new}}}{M_{c}(s')} \geq
F(c') \frac{m_{p_{new}}}{M_{c'}(s') + m_{p_{new}}}$, so 
$p_{new}$ is stable in $s'$.
Next, consider a miner $p$ s.t. $s.p=c' \neq c$.
Since $s$ is stable, we know that $\forall c'' \in C$, $F(c')
\frac{m_{p}}{M_{c'}(s)} \geq F(c'')
\frac{m_{p}}{M_{c''}(s) + m_{p}}$.
Now since $M_{c}(s) < M_{c}(s')$ and for all $c'' \neq
c$ $M_{c''}(s) = M_{c''}(s')$, we get that $\forall c''
\in C, F(c') \frac{m_{p}}{M_{c'}(s')} \geq F(c'')
\frac{m_{p}}{M_{c''}(s') + m_{p}}$.
Meaning that $p$ is stable in $s'$.

Finally, consider $p \neq p_{new}$ s.t.\ $s'.p = c$. 
We need to show that for all $c' \neq c$, \ $F(c')
\frac{m_{p}}{M_{c'}(s') + m_{p}} \leq F(c)
\frac{m_{p}}{M_{c}(s')}$.
Again, since we pick $c$ to be $\underset{c' \in C} \argmax
~F(c') \frac{m_{p_{new}}}{M_{c_j}(s) + m_{p_{new}}}$, we
know that $\forall c' \neq c$, $F(c)
\frac{m_{p_{new}}}{M_{c}(s) + m_{p_{new}}} \geq F(c')
\frac{m_{p_{new}}}{M_{c'}(s) + m_{p_{new}}} $, and thus
$\forall c' \neq c$, $F(c) \frac{1}{M_{c}(s')} \geq F(c')
\frac{1}{M_{c'}(s') + m_{p_{new}}}$.
Now by the claim assumption, $m_{p_{new}} \leq m_{p}$.
Therefore, $\forall c' \neq c$, $F(c) \frac{1}{M_{c}(s')}
\geq F(c')
\frac{1}{M_{c'}(s') + m_{p_{new}}} \geq F(c')
\frac{1}{M_{c'}(s') + m_{p}}$.
Thus, $\forall c' \neq c$, $F(c) \frac{m_{p}}{M_{c}(s')}
\geq F(c') \frac{m_{p}}{M_{c'}(s') + m_{p}}$.
Meaning that $p$ is stable in $s'$.

\end{proof}


\begin{proposition}
\label{lem:eq}

For any set of of miners $\Pi = \{p_1,\ldots,p_n\}$, set of
coins $C$, and a reward function $F$, the game
$G_{\Pi,C,F}$ has a stable configuration.

\end{proposition}

\begin{proof}

Order miners by mining power so that $m_{p_1} \geq m_{p_2}
\geq \ldots \geq m_{p_n}$.
First we show that the game $G_{\{p_1\},C,F}$
has a stable configuration.
Let $c = \underset{c' \in C} \argmax ~ F(c')$, and
define a configuration $s$ in $G_{\{p_1\},C,F}$ as follows:
$s.p_1 = c$.
Since $\forall c' \in C,~ F(c) \frac{m_{p_1}}{M_{c}(s)}
= F(c) \geq F(c') = F(c') \frac{m_{p_1}}{M_{c'}(s) +
m_{p_1}}$, we get that the configuration $s$ is stable in
$G_{\{p_1\},C,F}$. 
The lemma follows by inductively applying Claim~\ref{claim:eq}.

\end{proof}

\section{Ordinal potential for the symmetric case}
\label{app:EasyPotential}

We consider here the symmetric case where all coin rewards are
equal, and show that any game in this case has a simple
potential function.

\begin{proposition}

Consider a finite set of players $\Pi$, a finite set of coins
$C$, and a reward function $F$ s.t.\ $\forall c,c' \in C$ $F(c) =
F(c')$.
Then $H(s) \triangleq  \mathlarger{\sum}_{c \in C}
\frac{1}{M_{c}(s)}$ is a potential function in the game
$G_{\Pi,C,F}$.

\end{proposition}

\begin{proof}

Consider two configurations $s,s' \in  S$ s.t.\ some better
response step of a miner $p \in \Pi$ leads from configuration
$s$ to configuration $s'$. 
We need to show that $H(s) > H(s')$.
Let $c = s.p$ and $c' =s'.p$.
By definition of a better response step,
$F(c')\frac{m_{p}}{M_{c'}(s) + m_{p}} > F(c)
\frac{m_{p}}{M_{c}(s)}$.
Since $F(c) = F(c')$, we get that $\frac{1}{M_{c'}(s)
+ m_{p}} >  \frac{1}{M_{c}(s)}$, and thus
\begin{align}
M_{c}(s) > M_{c'}(s) + m_{p}. \label{eq:sp}
\end{align}
Now since for every $c'' \notin \{c', c\}$, $M_{c''}(s) =
M_{c''}(s')$, we get that 
\[H(s') - H(s) = \frac{1}{M_{c'}(s')} +
\frac{1}{M_{c}(s')} - (\frac{1}{M_{c'}(s)} +
\frac{1}{M_{c}(s)}).\]
Moreover, since $p$ moves from $c$ to $c'$ in $s$, we get that
\[H(s') - H(s) = \frac{1}{M_{c'}(s) +
m_{p}} + \frac{1}{M_{c}(s) - m_{p}} -
(\frac{1}{M_{c'}(s)} + \frac{1}{M_{c}(s)}).\]
Thus, in order to show that $H(s) >
H(s')$, we need to show that 
\[\frac{1}{M_{c'}(s) +
m_{p}} + \frac{1}{M_{c}(s) - m_{p}} <
\frac{1}{M_{c'}(s)} + \frac{1}{M_{c}(s)}.\]
Equivalently, that
 \[\frac{M_{c}(s) - m_{p} + m_{p} +
M_{c'}(s)}{(M_{c'}(s) + m_{p})(M_{c}(s) - m_{p})}
< \frac{M_{c}(s) + M_{c'}(s)}{M_{c}(s)M_{c'}(s)}, \]
or that 
\[M_{c}(s)M_{c'}(s) < (M_{c'}(s) +
m_{p})(M_{c}(s) - m_{p}),\] 
or finally that
\[M_{c}(s)(M_{c'}(s) + m_{p}) - M_{c}(s)m_{p}
< (M_{c'}(s) + m_{p})M_{c}(s) - (M_{c'}(s) +
m_{p})m_{p}).\]
The proposition follows from Equation~\ref{eq:sp} 

\end{proof}

\section{Observations' proofs for the ordinal potential}
\label{app:ordinalPotential}

\begin{obsclone}{ob:neverMoveBack}
Consider a game $G_{\Pi,C,F}$, $s \in S$, 
$v_i(s) \in C$, and $p \in \Pi$ s.t.\ $s.p=v_i(s)$.
Then in every better response step of $p$ that changes $s.p$ to
a coin $v_j(s)$, it  that $j > i$.

\end{obsclone}

\begin{proof}

By the definition of a better response step, 
$\frac{F(v_j(s))}{M_{v_j(s)}(s) + m_{p}} >
\frac{F(v_i(s))}{M_{v_i(s)}(s)}$, and thus
$RPU_{v_j(s)}(s) = \frac{F(v_j(s))}{M_{v_j(s)}(s)} >
\frac{F(v_i(s))}{M_{v_i(s)}(s)} = RPU_{v_i(s)}(s)$.
By definition of $v(s)$, we get that $j >i$. 

\end{proof}

\begin{obsclone}{ob:movingImproving}
Consider a game $G_{\Pi,C,F}$.
If some better response step from configuration $s$ to
configuration $s'$ of a miner $p$ changes $s.p = c$ to
$s'.p = c'$, then $RPU_c(s) < min(RPU_c(s'), RPU_{c'}(s'))$.

\end{obsclone}

\begin{proof}

By definition of a better response step, $RPU_{c'}(s')
> RPU_{c}(s)$.
In addition, since $M_{c}(s') = M_{c}(s) - m_{p}$, we get
that $RPU_{c}(s') = \frac{F(c)}{M_{c}(s')} =
\frac{F(c)}{M_{c}(s) - m_{p}}  > \frac{F(c)}{M_{c}(s)} =
RPU_{c}(s)$.

\end{proof}

\section{There is often a better Eeuilibrium: proofs}
\label{app:oftenBettr}

We prove here the Claims from
Section~\ref{sec:betterEquilibrium}

\begin{claimclone}{claim:higherpayoff}
Consider a game $G_{\Pi,C,F}$ under Assumption~\ref{ass:genric}. 
If the game has more than one stable configuration,
then for every stable configuration $s$ there exist a
miner $p$ and a stable configuration $s'$ s.t.\ $u_{p}(s') >
u_{p}(s)$.

\end{claimclone}

\begin{proof}

Consider a stable configuration $s$.
By assumption, there exists another stable
configuration $s' \neq s$.
Therefore, there is a player $p$ and coins
$c \neq c'$ s.t.\ $p \in P_{c}(s)$ and $p \in
P_{c'}(s')$.
By Assumption~\ref{ass:genric}, $\frac{F(c)}{M_{c}(s)}
\neq \frac{F(c')}{M_{c'}(s')}$.
Thus, $u_{p}(s') \neq u_{p}(s)$.
If $u_{p}(s') > u_{p}(s)$, then we are done.
Otherwise, by Observation~\ref{obs:optimal},
there is another player $p' \neq p$ s.t.\ $u_{p'}(s') >
u_{p'}(s)$.

\end{proof}

\begin{claim}
\label{claim:SSTBS}
Consider a game $G_{\Pi,C,F}$, a configuration $s \in S$, a coin
$c \in C$, and two miners $p,p' \in P_{c}(s)$ s.t.\ $m_{p} \leq
m_{p'}$.
If $p$ is stable in $s$, then $p'$ is stable in $s$ as well.

\end{claim}

\begin{proof}
  
  Since $p$ is stable in $s$,
  we get that $F(c) \frac{m_{p}}{M_{c}(s)} \geq
  F(c') \frac{m_{p}}{M_{c'}(s) + m_{p}} $ for every $c'
  \in C$.
  Thus, $F(c) \frac{1}{M_{c}(s)} \geq F(c')
  \frac{1}{M_{c'}(s) + m_{p}}$ for every $c' \in C$.
  Since, $m_{p} \leq m_{p'}$, it follows that $F(c)
  \frac{1}{M_{c}(s)}  \geq F(c') \frac{1}{M_{c'}(s) + m_{p'}}$ 
  for every $c' \in C$.
  Thus, $F(c) \frac{m_{p'}}{M_{c}(s)} \geq F(c')
  \frac{m_{p'}}{M_{c'}(s) + m_{p'}}$ for every $c' \in C$,
  and $p'$ is stable in $s$.

\end{proof}

\begin{claimclone}{claim:SRS}
Let $F$ be a reward function.
Consider a system $Q= \langle \Pi,C \rangle$, and a
configuration $s \in S_Q$.
Now consider another system $Q'= \langle \Pi',C \rangle$ s.t.\
$\Pi' = \Pi \cup \{p_{new}\}$, $p_{new} \not\in \Pi$, and
$m_{p_{new}} \leq min\{m_{p} | p \in \Pi \}$.
Let $c = \underset{c' \in C} \argmax ~F(c')
\frac{m_{p_{new}}}{M_{c'}(s) + m_{p_{new}}}$ and consider a
configuration $s' \in S_{Q'}$ s.t.\ for all $p \in \Pi$
$s'.p=s.p$ and $s'.p_{new} = c$.
Then $p_{new}$ is stable in $s'$ in game $G_{\Pi',C,F}$, and
every player $p \in \Pi$ that is stable in $s$ in $G_{\Pi,C,F}$
is also stable in $s'$ in $G_{\Pi',C,F}$.
\end{claimclone}

\begin{proof}

By construction of configuration $s'$, $M_{c}(s')= M_{c}(s)
+ m_{p_{new}}$ and $\forall c' \neq c$ $M_{c'}(s') =
M_{c'}(s)$.
Therefore, by the way we pick $c$, we get that $F(c)
\frac{m_{p_{new}}}{M_{c}(s')} \geq F(c')
\frac{m_{p_{new}}}{M_{c'}(s')}$ for all $c' \in C$, and thus  
$p_{new}$ is stable in $s'$.
Now consider a player $p \in \Pi$ that is stable in $s$, we show
that $p$ is stable also in $s'$.
Consider two cases:
\begin{itemize}
  
  \item First, $p \in P_{c'}(s')$ s.t.\ $c \neq c'$.
  By construction, $p \in P_{c'}(s)$, and
  since $p$ is stable in $s$ we know that $F(c')
  \frac{m_{p}}{M_{c'}(s)} \geq F(c'')
  \frac{m_{p}}{M_{c''}(s) + m_{p}}$ for every $c'' \in C$.
  Now since $M_{c}(s') > M_{c}(s)$ and for all $c'' \neq c$,
  $M_{c''}(s') = M_{c''}(s)$, we get $F(c')
  \frac{m_{p}}{M_{c'}(s')} \geq F(c'')
  \frac{m_{p}}{M_{c''}(s') + m_{p}}$ for every $c'' \in C$.
  Meaning that $p$ is stable in $s'$.

  \item Second, $p \in P_{c}(s')$.
  Since $p_{new}$ is stable in $s'$ and $m_p \geq
  m_{p_{new}}$, we get by Claim~\ref{claim:SSTBS} that $p$ is
  stable in $s'$.
  
\end{itemize}

\end{proof}

\begin{lemma}
\label{claim:eqNotEnique}

Any game $G_{\Pi,C,F}$ under Assumptions~\ref{ass:alone}
and~\ref{ass:genric} has at least two different stable
configurations.

\end{lemma}

\begin{proof}

Let $p_1,\ldots,p_n$ be the miners in $\Pi$ sorted in decreasing
mining power, i.e., $m_{p_1} \geq m_{p_2} \geq
\ldots \geq m_{p_n}$ and let $c_1,\ldots,c_l$ be the coins in $C$
sorted in decreasing coin rewards, i.e., $F(c_1) \geq F(c_2)
\geq \ldots \geq F(c_l)$.
Note that through the proof we construct several games, but we
assume Assumptions~\ref{ass:alone} and~\ref{ass:genric} only in
the game $G_{\Pi,C,F}$.
Let $\Pi_1,\ldots,\Pi_n$ be a sequence of sets of miners
s.t.\ $\Pi_k = \{p_1,\ldots,p_k\}$, $1 \leq k \leq n$.
Next consider two configurations $s^2_1,s^2_2$ in game
$G_{\Pi_2,C,F}$: $s^2_1 = \langle c_1, c_2 \rangle$ (i.e.,
$s^2_1.p_1 = c_1$ and $s^2_1.p_2 = c_2$) and $s^2_2 = \langle
c_2, c_1 \rangle$.
Note that $s^2_1 \neq s^2_2$, and since $F(c_1) \geq F(c_2)$,
$p_1$ is stable in $s^2_1$ and $p_2$ is stable in $s^2_2$.

We now use $s^2_1, s^2_2$ to inductively construct a sequence
of tuples of configurations\\ $\langle s^2_1, s^2_2 \rangle,
\langle s^3_1, s^3_2 \rangle,\ldots, \langle s^n_1, s^n_2
\rangle$, where $\forall ~ 2 \leq k \leq n$, $s^k_1, s^k_2$ are
two configurations in Game $G_{\Pi_k,C,F}$.
For every $ 3 \leq k \leq n$ let $x_1 = \underset{c \in C}
\argmax ~F(c) \frac{m_{p_k}}{M_{c}(s^{k-1}_1) + m_{p_k}}$
and  $x_2 = \underset{c \in C}
\argmax ~F(c) \frac{m_{p_k}}{M_{c}(s^{k-1}_2) + m_{p_k}}$,
we construct $s^k_1 = s^{k-1}_1 \times x_1$ (i.e., $\forall
1 \leq i < k$, $s^k_1.p_i = s^{k-1}_1.p_i$ and $s^k_1.p_k = x_1$)
and $s^k_2 = s^{k-1}_2 \times x_2$.

Note that since $s^2_1 \neq s^2_2$, we get by construction that
$s^n_1 \neq s^n_2$.
It remains to show that configurations $s^n_1$ and $s^n_2$
are stable.
Assume by way of contradiction that it is not the case, and
assume w.l.o.g. that $s^n_1$ is not stable.
By inductively applying Claim~\ref{claim:SRS}, and since $p_1$
is stable in $s^2_1$ we get that players $p_1$ and
$p_3,\ldots,p_n$ are stable in $s^n_1$. 
Thus $p_2$ is not stable.
Recall that $s^2_1.p_2 = c_2$, and thus by construction
$s^n_1.p_2 = c_2$.
Recall also that we assume Assumptions~\ref{ass:alone}
and~\ref{ass:genric}, and consider two cases:
\begin{itemize}

  \item First, $|P_{c_2}(s^n_1)| = 1$ . 
  By Assumption~\ref{ass:alone}, there is a miner $p \in \Pi$
  s.t.\ changing $s^n_1.p$ to $c_2$ is a better response step
  for $p$.
  Thus, $p$ is not stable in $s^n_1$.
  In addition, by definition of better response step, we know
  that $p_i \neq p_2$.
  A contradiction to $p_2$ being the only not stable miner in
  $s^n_1$.
  

  \item Second, $|P_{c_2}(s^n_1)| > 1$.
  Since $p_1 \in P_{c_1}(s^2_1)$, we get that there is a stable 
  miner $p \in P_{c_2}(s^2_1)$ s.t.\ $p \in
  \{p_3,\ldots,p_n\}$,
  and thus, $m_2 \geq m_{p}$.
  Therefore, by Claim~\ref{claim:SSTBS}, we get that $p_2$ is
  stable in $s^n_1$.
  A contradiction.
  
\end{itemize}

\end{proof}

\section{Reward design: proof of Lemma~\ref{lem:inv}}
\label{app:rewardDesign}

\begin{lemmaclone}{lem:inv}
Consider a configuration $s \in T_i \setminus \{s_i\}$.
Then every better response learning in the game
$G_{\Pi,C,H_i(s)}$ that starts at $s$ converges to a
configuration $s' \in T_i$ such that:
\begin{enumerate}
  
  \item $\forall k, ~1 \leq k < m_i(s)$, $s'.p_k = s.p_k$. 
  
  \item  $s'.p_{m_i(s)} = s_f.p_i$.
  
\end{enumerate}

\end{lemmaclone}

\begin{proof}

Let $c = s_f.p_{i-1}$ and $c'=s_f.p_{i}$.
By definition of $m_i(s)$, and since $s \neq s_i$ we know
that $p_{m_i(s)} \in P_{c}(s)$.
We first show that the only better response step in
configuration $s$ in game $G_{\Pi,C,H_i(s)}$ is that
$p_{m_i(s)}$ moves to $c'$.
Since $\forall c'' \neq c'$, $RPU_{c''}(s) = R(s)$ and
$RPU_{c'}(s) > R(s)$, we get that no miner has a better response
step to move to any coin $c'' \neq c'$.
Since $m_{p_{m_i(s)}} < m_{p_{a_i(s)}}$, we get that
$RPU_{c'}((s_{-p_{m_i(s)}}, c')) =
\frac{R(s) \cdot (M_{c'}(s) + m_{p_{a_i(s)}})}{M_{c'}(s) +
m_{p_{m_i(s)}}} > R(s) = RPU_{c}(s)$, and thus moving to $c'$ is
a better response step for $p_{m_i(s)}$ in $s$.
Now consider a miner $p_k \not\in P_{c'(s)}$ s.t.\ $k \neq
m_i(s)$.
By definition of $m_i(s)$, $k< m_i(s)$, thus $m_{p_k} >
m_{p_{m_i(s)}}$, and thus, $\frac{R(s) \cdot (M_{c'}(s) +
m_{p_{a_i(s)}})}{M_{c'}(s) + m_{p_{k}}} \leq R(s)$.
All in all, we get that the only better response step in
configuration $s$ is that $p_{m_i(s)}$ moves to $c'$.

Let $s^0 = (s_{-p_{m_i(s)}},c')$ be the configuration reached
after $p_{m_i(s)}$ takes its step.
We next prove by induction that every configuration $s'$ that is
reached by better response steps staring from $s^0$ 
satisfies the following properties:


\begin{itemize}

  \item[$\Psi_1$.] $\forall k, ~1 \leq k < m_i(s)$: $s'.p_k =
  s.p_k$.
  \label{en:p2}
  
    \item[$\Psi_2$.]  $s'.p_{m_i(s)} = c'$.
   
  \label{en:p1}
  
  \item[$\Psi_3$.] $\forall k, ~m_i(s) < k \leq n$: $s'.p_k \in
  \{c, c'\}$.
  \label{en:p3}
  
%
%
  
  \item[$\Psi_4$.]$M_c(s^0) \leq M_{c}(s') \leq M_c(s) $.
  
  \label{en:p7}
  
   \item[$\Psi_5$.]$M_{c'}(s) \leq M_{c'}(s')
   \leq M_{c'}(s^0)$.
   
  \label{en:p8}
\end{itemize}

\noindent Note that since $s \in T_i$, $\Psi_1 - \Psi_3$ imply
that $s' \in T_i$.

\textbf{Base.} We show below that the properties $\Psi_{1} -
\Psi_{5}$ are satisfied for configuration $s^0$.

\begin{itemize}
  
  \item[$\Psi_1 - \Psi_3$.] Follow by definitions of
  $m_i(s)$ and $T_i$, and since $s^0 =(s_{-p_{m_i(s)}},c')$.
  
  \item[$\Psi_4 - \Psi_5$.] Trivially follows.
  
%
  
  
%

  
%
%
   
\end{itemize}

\textbf{Induction step.} Consider two configurations $s^1,s^2
\in S$ s.t.\ a better response step of some miner leads
from $s^1$ to $s^2$, and $s^1$ satisfies properties
$\Psi_{1} - \Psi_{5}$.
We show that $s_2$ satisfies properties $\Psi_{1} -
\Psi_{5}$ as well.

\paragraph{Useful equations for configuration $s^1$.}
We start by proving two useful equations on the RPUs of coins $c$
and $c'$ in configuration $s^1$ using $\Psi_{4}$ and
$\Psi_{5}$:

By $\Psi_{4}$, $M_{c}(s^1) \leq M_{c}(s) $.
In addition, since $H_i(s)(c) = R(s)\cdot M_{c}(s)$, we get that
\begin{equation} 
\label{eq:RPUc}
RPU_{c}(s^1) =
\frac{H_i(s)(c)}{M_{c}(s^1)}= \frac{R(s)\cdot
M_{c}(s)}{M_{c}(s^1)} \geq \frac{R(s) \cdot M_{c}(s)}{M_{c}(s)} = R(s).
\end{equation}
By $\Psi_5$, $M_{c'}(s^1) \leq M_{c'}(s^0)$. 
In addition, since (1)
$m_{p_{m_i(s)}} < m_{p_{a_i(s)}}$, (2) $M_{c'}(s)=
M_{c'}(s^0) - m_{p_{m_i(s)}}$, and (3) $H_i(s)(c') = R(s)\cdot
(M_{c'}(s) + m_{p_{a_i(s)}})$, we get that
\begin{equation}
\label{eq:RPUc'}
\begin{split}
     RPU_{c'}(s^1) &= \frac{H_i(s)(c')}{M_{c'}(s^1)}=
     \frac{R(s) \cdot (M_{c'}(s) + m_{p_{a_i(s)}})}{M_{c'}(s^1)}
     \\
     & =\frac{R(s)\cdot (M_{c'}(s^0) - m_{p_{m_i(s)}} +
     m_{p_{a_i(s)}})}{M_{c'}(s^1)}  \\
     &\geq
     \frac{R(s)\cdot M_{c'}(s^0)}{M_{c'}(s^0)} = R(s).
\end{split}
\end{equation}

\paragraph{Useful equations for configurations $s^1$ and $s^2$.}
We now prove two useful equations on the RPUs of coins not in
$\{c,c'\}$ in configurations that satisfy $\Psi_{1} -
\Psi_{3}$.

By definition of $H_i$, $\forall c'' \not\in \{c,c'\}$,
$H_i(s)(c'')= R(s)\cdot M_{c''}(s)$.
Since $ s \in T_i$, $\forall k, i \leq k \leq n: ~s.p_k
\in \{c,c'\}$, and by definition of $m_i(s)$, we know that
$m_i(s) \geq i$.
Therefore, by $\Psi_1 - \Psi_3$, we get that
\begin{equation}
\label{eq:Mc''}
\forall c'' \not\in \{c,c'\}, ~M_{c''}(s^1) = M_{c''}(s),
\end{equation}
and thus
\begin{equation}
\label{eq:RPUc''}
\forall c'' \not\in \{c,c'\}, ~RPU_{c''}(s^1)
=\frac{H_i(s)(c'')}{M_{c''}(s^1)} 
=\frac{R(s)\cdot M_{c''}(s)}{M_{c''}(s)}
= R(s).
\end{equation}

\paragraph{Induction step proof.} We are now ready to prove that
$s_2$ satisfies properties $\Psi_{1} - \Psi_{5}$.

\begin{itemize}
  \item[$\Psi_1 - \Psi_2$.] Since $\Psi_1$ and $\Psi_2$ hold in
  $s^1$, we only need to show that $p_1,\ldots,p_{m_i(s)}$ are
  stable in $s^1$.
  Let $k \in \{1,\ldots,m_i(s)\}$ and note that $m_{p_k} \geq
  m_{p_{m_i(s)}}$.
  By Equation~\ref{eq:RPUc''}, for all $c'' \not\in \{c,c'\}$,
  $RPU_{c''}(s^1) = R(s)$, and by Equations~\ref{eq:RPUc}
  and~\ref{eq:RPUc'}, we know that $RPU_{c}(s^1) \geq R(s)$ and
  $RPU_{c'}(s^1) \geq R(s)$, respectively.
  Therefore, it remains to show that $p_k$ does not have a
  better response step to $c$ or $c'$:
  \begin{itemize}
    
    \item  By~$\Psi_{4}$, $M_{c}(s^1) \geq M_{c}(s^0)$, and thus
    we get that \[\frac{H_i(c)}{M_{c}(s^1) +
    m_{p_k}} \leq \frac{H_i(c)}{M_{c}(s^0) + m_{p_{m_i(s)}}} =
    \frac{H_i(c)}{M_{c}(s)} = \frac{M_{c}(s)\cdot R(s)}{M_{c}(s)}
    = R(s).
    \]
    Therefore, $p_k$ does not have a better response step to $c$. 
    
    \item By $\Psi_2$, $p_{m_i(s)} \in P_{c'}(s^1)$.
    Therefore, if $p_k = p_{m_i(s)}$, then we are done.
    Otherwise, $k < m_i(s) = a_i(s) + 1$, so $m_{p_k}
    \geq m_{p_{a_i(s)}}$.
    By $\Psi_{5}$, $M_{c'}(s^1) \geq M_{c'}(s)$, and thus
    \[
    \frac{H_i(c')}{M_{c'}(s^1) + m_{p_k}} \leq
    \frac{H_i(c')}{M_{c'}(s) + m_{p_{a_i(s)}}}
    =  \frac{ R(s) \cdot (M_{c'}(s) + m_{p_{a_i(s)}})}{M_{c'}(s)
    + m_{p_{a_i(s)}}} = R(s).
    \]
    Therefore, $p_k$ does not have a better response step to
    $c'$.
  \end{itemize}  

  \item[$\Psi_3$.]
  By Equations~\ref{eq:RPUc} and~\ref{eq:RPUc'},
  $RPU_{c}(s^1) \geq R(s)$ and $RPU_{c'}(s^1) \geq R(s)$.
  By Equation~\ref{eq:RPUc''}, 
  for all $c'' \not\in \{c, c'\}$, $RPU_{c''}(s^1) =
  R(s)$.
  Therefore, since by the inductive assumption ($\Psi_3$), miners
  $ p_{m_i(s)+1},\ldots,p_n$ are in $P_{c}(s^1) \cup
  P_{c'}(s^1)$, we get that none of them has a better response
  step to move to a coin $c'' \not\in \{c,c'\}$, and thus
  $\Psi_3$ holds in $s^2$ as well.
  
%
%
  
  \item[$\Psi_4$.] 
  By definitions of $s^0$, $m_i(s)$, and $T_i$, we know that
  $P_{c}(s^0) \cap \{p_{m_i(s)},\ldots,p_n\} = \emptyset$.
  Since $\Psi_1$ holds in $s^0$ and in $s^1$, we get that $\forall
  k, ~1 \leq k < m_i(s)$: $s^1.p_k = s^0.p_k$.
  Therefore, we get that $M_c(s^0) \leq M_{c}(s^2)$. 
  It remains to show that $M_{c}(s^{2}) \leq
  M_{c}(s)$.
  By Equation~\ref{eq:RPUc''}, for all $c'' \not\in \{c,c'\}$,
  $RPU_{c''}(s^1) = R(s)$.
  By Equations~\ref{eq:RPUc} and~\ref{eq:RPUc'}, 
  $RPU_{c}(s^1) \geq R(s)$ and $RPU_{c'}(s^1) \geq R(s)$.
  Now assume by way of contradiction that $M_{c}(s^{2}) >
  M_{c}(s)$.
  Thus, 
    \[
   RPU_{c}(s^2) =
   \frac{H_i(c)}{M_{c}(s^2)} <
   \frac{ H_i(c)}{M_{c}(s)} =
   \frac{ R(s) \cdot M_{c}(s)}{M_{c}(s)} 
   = R(s).
  \]
  Now since $RPU_{c}(s^2) < R(s) \leq RPU_{c}(s^1)$, we get that
  $s^2 =(s^1_{-p},c)$ for some $p \in \Pi$.
  A contradiction to Observation~\ref{ob:movingImproving}. 
  
  \item[$\Psi_5$.] Since we already showed that $s^2$ satisfies
  $\Psi_1 - \Psi_3$, by Equation~\ref{eq:Mc''}, we get that
  $\forall c'' \not\in \{c,c'\}, ~M_{c''}(s^2) = M_{c''}(s)$.
  In addition, since $s^0 =(s_{-p_{m_i(s)}},c')$, $\forall c''
  \not\in \{c,c'\}, ~M_{c''}(s^0) = M_{c''}(s)$.
  Therefore, we get that $\forall c''
  \not\in \{c,c'\}, ~M_{c''}(s) = M_{c''}(s^0) = M_{c''}(s^2)$,
  and thus $M_{c}(s^2) + M_{c'}(s^2) = M_{c}(s^0)
  + M_{c'}(s^0) = M_{c}(s) + M_{c'}(s)$. 
  Hence, $\Psi_5$ follows from $\Psi_4$.
  
\end{itemize}

Now together with Theorem~\ref{theorm:generalPotential}, we
know that every better response learning in the game
$G_{\Pi,C,H_i(s)}$ that starts at $s$ converges to some
configuration $s'$ that satisfies $\Psi_1 - \Psi_5$.
Since $s \in T_i$, we get that
\[ (\forall k, ~1 \leq k \leq i-1~:
s.p_k=s_f.p_k) \wedge (\forall k, i \leq k \leq n:~ s.p_k \in
\{s_f.p_i, s_f.p_{i-1}\}).
\]
And since $s \neq s_i$, we get that $i \leq m_i(s)$.
Thus, by $\Psi_1$, $\forall k, ~1 \leq k \leq i-1 :
s'.p_k=s_f.p_k$.
In addition, by $\Psi_3$, we get that
$\forall k, ~i \leq k \leq n, ~s.p_k \in \{s_f.p_i,
s_f.p_{i-1}\}$.
Therefore, $s' \in T_i$, and the lemma follows from $\Psi_1$ and
$\Psi_2$.

\end{proof}